\documentclass{article}
\pdfoutput=1
\usepackage[latin1]{inputenc}
\usepackage{enumerate,ifthen,xspace}
\usepackage{graphicx}
\usepackage{amsmath,amsfonts,amstext,amssymb,stmaryrd,amsthm}
\usepackage{hyperref}

\long\def\beginpgfgraphicnamed#1#2\endpgfgraphicnamed{\includegraphics{#1}}

\DeclareMathOperator{\prefto}{\rhd}
\DeclareMathOperator{\nprefto}{\not\kern-2.5pt\rhd}

\newcommand{\mso}{\dot\{}
\newcommand{\msc}{\dot\}}

\newcommand{\defect}{\ensuremath{\mathbb{D}}}
\newcommand{\mcup}{\mathop{\dot\cup}}

\newtheorem{theorem}{Theorem}
\newtheorem{lemma}[theorem]{Lemma}

\newtheorem{note}[theorem]{Note}

\theoremstyle{definition}
\newtheorem{definition}[theorem]{Definition}
\newtheorem{example}[theorem]{Example}

\newcommand{\ES}{\mbox{$\emptyset$}}

\newcommand{\A}{\mbox{$\ \wedge\ $}}

\newcommand{\sse}{\mbox{$\:\subseteq\:$}}

\newcommand{\fa}{\mbox{$\forall$}}
\newcommand{\te}{\mbox{$\exists$}}

\newcommand{\LL}{\mbox{$\ldots$}}

\newcommand{\C}[1]{\mbox{$\{{#1}\}$}}           

\newcommand{\NI}{\noindent}
\newcommand{\HB}{\hfill{$\Box$}}
\newcommand{\VV}{\vspace{5 mm}}

\newcommand{\II}{\vspace{2 mm}}

\newcommand{\szkew}[1]{\relax \setbox0=\hbox{\kern -24pt $\displaystyle#1$\kern 0pt }%
\box0}
{\catcode`\@=11 \global\let\ifjusthvtest@=\iffalse}

\newcounter{oldmycaption}

\title{A Generic Approach to Coalition Formation}
\author{Krzysztof R. Apt and Andreas Witzel\thanks{Andreas Witzel is supported by
    a GLoRiClass fellowship funded by the European Commission
    (Early Stage Research Training Mono-Host Fellowship MEST-CT-2005-020841).} \\
{\small CWI, Amsterdam, the Netherlands} \\
{\small and University of Amsterdam}
}
\date{}

\ifx\pdfoutput\undefined\else\fi

\begin{document}
\sloppy

\maketitle

\begin{abstract}
  We propose an abstract approach to coalition formation that focuses
  on simple merge and split rules transforming partitions of a group
  of players. We identify conditions under which every iteration of
  these rules yields a unique partition. The main conceptual tool is
  a specific notion of a stable partition. The results are parametrized by a
  preference relation between partitions of a group of players and
  naturally apply to coalitional TU-games, hedonic games and exchange
  economy games.
\end{abstract}



\section{Introduction}

\subsection{Approach}

Coalition formation has been a research topic of continuing interest
in the area of coalitional games.  It has been analyzed from several
points of view, starting with \cite{AD74}, where the static situation
of coalitional games in the presence of a given coalition structure
(i.e.~a partition) was considered. 

In this paper we consider the perennial question `how do coalitions
form?' by proposing a simple answer: `by means of merges and splits'.
This brings us to the study of a natural problem, namely under what
assumptions the outcomes of arbitrary sequences of merges and splits
are unique.

These considerations yield an abstract approach to coalition formation
that focuses on partial preference relations between partitions of a
group of players and simple merge and split rules. These rules
transform partitions of a group of players under the condition that
the resulting partition is preferred. By identifying conditions under
which every iteration of these rules yields a unique partition we are
brought to a natural notion of a stable partition.

This approach is parametrized by a generic preference
relation. The obtained results depend only on a few simple properties,
namely irreflexivity, transitivity and monotonicity, and do not require
any specific model of coalitional games.

In the case of coalitional TU-games the preference relations induced
by various well-known orders on sequences of reals, such as leximin or
Nash order, satisfy the required properties. As a consequence our
results apply to the resulting preference relations and coalitional
TU-games.  We also explain how our results apply to hedonic games
(games in which each player has a preference relation on the sets of
players that include him) and exchange economy games.

This approach to coalition formation is indirectly inspired by the
theory of abstract reduction systems (ARS), see, e.g.~\cite{Ter03}, one
of the aims of which is a study of conditions that guarantee a unique
outcome of rule iterations.  In an earlier work \cite{Apt04} we
exemplified another benefit of relying on ARS by using a specific
result, called Newman's Lemma, to provide uniform proofs of order
independence for various strategy elimination procedures for finite
strategic games.  

\subsection{Related work}

Because of this different starting point underpinning our approach, it
is difficult to compare it to the vast literature on the subject
of coalition formation.  Still, a number of papers should be mentioned
even though their results have no bearing on ours.

In particular, rules that modify coalitions are considered in \cite{Yi97} in the
presence of externalities and in \cite{RV97} in the presence of
binding agreements.  In both papers two-stage games are analyzed. In
the first stage coalitions form and in the second stage the players
engage in a non-cooperative game given the emerged coalition
structure.  In this context the question of stability of the coalition
structure is then analyzed.

The question of (appropriately defined)
stable coalition structures often focused on hedonic
games.  \cite{BJ02} considered four forms of
stability in such games: core, Nash, individual and
contractually individual stability.  Each alternative captures the
idea that no player, respectively, no group of players has an incentive
to change the existing coalition structure.  The problem of existence
of (core, Nash, individually and contractually individually) stable
coalitions was considered in this and other
references, for example \cite{SBK01} and~\cite{BZ03}.  

Recently, \cite{BJ06} compared various notions of stability
and equilibria in network formation games. These are games in which the
players may be involved in a network relationship that, as a graph, may evolve.
Other interaction structures which players can form were considered in
\cite{Dem04}, in which formation of hierarchies was studied, and
\cite{MSPCP06}, in which only bilateral agreements that follow a
specific protocol were allowed.

Early research on the subject of coalition formation is discussed in
\cite{Gre94}.  More recently, various aspects of coalition formation
are discussed in the collection of articles \cite{DW05} and in 
the survey \cite{RePEc:urb:wpaper:07_12}.
Initially, we obtained the corresponding results in
\cite{AR06} in a limited setting of coalitional TU-games and the
preference relation induced by the utilitarian order.



\subsection{Plan of the paper}

The paper is organized as follows. In the next section we set the
stage by introducing an abstract comparison relation between
partitions of a group of players and the corresponding merge and split
rules that act on such partitions.  Then in Section \ref{sec:TU} we
discuss a number of natural comparison relations on partitions
within the context of coalitional TU-games and in Section \ref{sec:ind}
by using arbitrary value functions for such games.

Next, in Section \ref{sec:stable}, we introduce and study a
parametrized concept of a stable partition and in Section
\ref{sec:unique} relate it to the merge and split rules. Finally, in
Section \ref{sec:conc} we explain how to apply the obtained results to specific coalitional
games, including TU-games, hedonic games and exchange economy games,
and in Section \ref{sec:conclusions} we summarize our approach.

\section{Comparing and transforming collections}

Let $N = \{1,2,\ldots ,n\}$ be a fixed set of players called the
\textit{grand coalition}.  Non-empty subsets of $N$ are called
\textit{coalitions}. A \textit{collection} (in the grand
coalition $N$) is any family $C := \{C_1,\ldots, C_l\}$ of mutually
disjoint coalitions, and $l$ is called its \textit{size}. If
additionally $\bigcup_{j=1}^l C_j = N$, the collection $C$ is called a
\textit{partition} of $N$.  For $C=\{C_1,\dots,C_k\}$, we define
$\bigcup C:=\bigcup_{i=1}^k C_i$.

In this article we are interested in comparing collections.
In what follows we only compare collections $A$ and $B$ that
are partitions of the same set, i.e.~such that $\bigcup A = \bigcup B$.
Intuitively, assuming a comparison relation $\prefto$, $A
\prefto B$ means that the way $A$ partitions $K$, where $K = \bigcup A
= \bigcup B$, is preferable to the way $B$ partitions $K$.

To keep the presentation uniform we only
assume that the relation $\prefto$ is irreflexive, 
i.e.~for no collection $A$, $A\prefto A$ holds, 
transitive, i.e.~for all
collections $A,B,C$
with $\bigcup A=\bigcup B=\bigcup C$,
$A\prefto B$ and $B\prefto C$ imply $A\prefto C$,
and that $\prefto$ is
\emph{monotonic} in the following two senses: for all collections $A,B,C,D$ with $\bigcup A=\bigcup
B$, $\bigcup C=\bigcup D$, and $\bigcup A\cap\bigcup C=\emptyset$,
\begin{equation*}
  \label{eq:mon1}\tag{m1}
  A\prefto B\text{ and }C\prefto D\text{ imply }A\cup C\prefto B\cup D,
\end{equation*}
and for all collections $A,B,C$ with $\bigcup A=\bigcup B$ and
$\bigcup A\cap\bigcup C=\emptyset$,
\begin{equation*}
  \label{eq:mon2}\tag{m2}
  A\prefto B\text{ implies }A\cup C\prefto B\cup C.
\end{equation*}

The role of monotonicity will become clear in Section \ref{sec:stable}, though property
\eqref{eq:mon2} will be already of use in this section.

\begin{definition}
By a \textit{comparison relation} we mean
an irreflexive and transitive relation on collections that satisfies the conditions
\eqref{eq:mon1} and \eqref{eq:mon2}.
\HB
\end{definition}

Each comparison relation $\prefto$ is used only to compare partitions
of the \emph{same} set of players.  So partitions of different sets of
players are incomparable w.r.t.~$\prefto$, that is no comparison
relation is linear. This leads to a more restricted form of linearity,
defined as follows.  We call a comparison relation $\prefto$ 
\emph{semi-linear} if for all collections $A,B$ with $\bigcup
A=\bigcup B$, either $A\prefto B$ or $B\prefto A$.

In what follows we study coalition formation by focusing on
the following two rules that allow us to transform partitions of 
the grand coalition:

\begin{description}
\item[merge:] $\{T_1,\dots,T_k\}\cup
  P\rightarrow\{\bigcup_{j=1}^kT_j\}\cup P$, where $\{\bigcup_{j=1}^kT_j\}\prefto\{T_1,\dots,T_k\}$
\item[split:] $\{\bigcup_{j=1}^kT_j\}\cup P\rightarrow\{T_1,\dots,T_k\}\cup
  P$, where $\{T_1,\dots,T_k\}\prefto\{\bigcup_{j=1}^kT_j\}$
\end{description}

Note that both rules use the $\prefto$ comparison relation `locally',
by focusing on the coalitions that take part and result from the merge, respectively split.
In this paper we are interested in finding conditions that guarantee that
arbitrary sequences of these two rules yield the same outcome.  
So, once these conditions hold, a specific \emph{preferred} partition exists 
such that any initial partition can be transformed into it by applying
the merge and split rules in an arbitrary order.

To start with, note that the termination of the iterations of these
two rules is guaranteed.

\begin{note}
\label{not:1}
Suppose that $\prefto$ is a comparison relation.
Then every iteration of the merge and split rules terminates.
\end{note}
\begin{proof}
Every iteration of these two rules produces by
\eqref{eq:mon2} a sequence of partitions
$P_1,P_2,\dots$ with $P_{i+1}\prefto P_i$ for all
  $i\geq1$. But the number of different partitions is
  finite. So by transitivity and irreflexivity of
  $\prefto$ such a sequence has to be finite.
\end{proof}

The analysis of the conditions guaranteeing the unique outcome of the
iterations is now deferred to Section \ref{sec:unique}.

\section{TU-games}
\label{sec:TU}

To properly motivate the subsequent considerations and to clarify the
status of the monotonicity conditions we now introduce some natural
comparison relations on collections for coalitional TU-games.  A
\emph{coalitional TU-game} is a pair $(N,v)$, where $N := \C{1,\LL,
  n}$ and $v$ is a function from the powerset of $N$ to the set of
non-negative reals\footnote{The assumption that the values of $v$ are
  non-negative is non-standard and is needed only to accomodate for
  the Nash order, defined below.} such that
$v(\emptyset) = 0$.

For a coalitional TU-game $(N,v)$ the comparison relations on collections
are induced in a canonic way from the corresponding comparison
relations on multisets of reals by stipulating that for the
collections $A$ and $B$

\begin{equation}
  \label{eq:prefto}
\mbox{$A \prefto B$ iff $v(A) \prefto v(B)$,}  
\end{equation}
where for a collection $A := \{A_1, \LL, A_m\}$,
$v(A) := \mso v(A_1), \LL, v(A_m)\msc$, denoting multisets in dotted braces.

So first we introduce the appropriate relations on multisets of non-negative reals.
The corresponding definition of monotonicity for such a relation $\prefto$
is that for all multisets $a,b,c,d$ of reals
\[
\mbox{$a \prefto b$ and $c \prefto d$ imply $a \mcup c \prefto b \mcup d$}
\]
and
\[
\mbox{$a \prefto b$ implies $a \mcup c \prefto b \mcup c$,}
\]
where $\mcup$ denotes multiset union.

Given two sequences $(a_1, \LL, a_m)$ and $(b_1, \LL, b_n)$ of real numbers
we define the (extended) \textit{lexicographic order} on them by
putting
\[
\mbox{$(a_1, \LL, a_m) >_{lex} (b_1, \LL, b_n)$}
\]
iff
\[
\mbox{$\te i \leq min(m,n) \: (a_i > b_i \A \fa j < i \: a_j = b_j)$}
\]
or
\[
\mbox{$\fa i \leq min(m,n) \: a_i = b_i \A m > n$.}
\]

Note that in this order we compare sequences of possibly different
length.  We have for example $(1,1,1,0) >_{lex} (1,1,0)$ and
$(1,1,0) >_{lex} (1,1)$. It is straightforward to check that it
is a linear order.

We assume below that $a = \mso a_1, \LL, a_m\msc$ and $b = \mso b_1, \LL,
b_n\msc$ and that $a^{*}$ is a sequence of the elements of $a$
in decreasing order, and define

\begin{itemize}

\item the \emph{utilitarian} order:

$a \succ_{ut} b$ iff $\sum_{i=1}^{m} a_i > \sum_{j=1}^{n} b_j$,

\item the \emph{Nash} order:

$a \succ_{Nash} b$ iff $\prod_{i=1}^{m} a_i > \prod_{j=1}^{n} b_j$,

\item the \emph{leximin} order:

$a \succ_{lex} b$ iff $a^{*} >_{lex} b^{*}$.

\end{itemize}

In \cite{Mou88} these relations were considered for sequences of the same
length. For such sequences we shall discuss in Section \ref{sec:ind}
two other natural orders. The intuition behind the Nash order is that
when the sum $\sum_{i=1}^{m} a_i$ is fixed, the product $\prod_{i=1}^{m}
a_i$ is largest when all $a_i$s are equal.  So in a sense the Nash
order favours an equal distribution.

The above relations are clearly irreflexive and transitive.
Additionally the following holds.

\begin{note}
The above three relations are all monotonic both in sense~\eqref{eq:mon1}
and~\eqref{eq:mon2}.
\end{note}
\begin{proof}
  The only relation for which the claim is not immediate is
$\succ_{lex}$. We only prove~\eqref{eq:mon1}
for $\succ_{lex}$; the remaining proof is analogous.

Let arbitrary multisets of non-negative reals
$a,b,c,d$ be given. We
define, with $e$ denoting any sequence or multiset of non-negative reals,
\begin{align*}
  len(e) &:=\text{the number of elements in $e$},\\
  \mu &:=\text{$(a\mcup b\mcup c\mcup d)^*$ with all duplicates removed},\\
  \nu(x,e)&:=\text{the number of occurrences of $x$ in $e$},\\
  \beta &:= 1+\max_{k=1}^{len(\mu)}\{\nu(\mu_k,a\mcup b\mcup c\mcup d)\},\\
  \#(e) &:= \sum_{k=1}^{len(\mu)}\nu(\mu_k,e)\cdot\beta^{-k}.
\end{align*}

So $\mu$ is the sequence of all distinct reals used in $a\mcup b\mcup c\mcup d$,
arranged in a decreasing order.
The function $\#(\cdot)$ injectively maps a multiset $e$ to a real number $y$ in such a way
that in the floating point representation of $y$ with base $\beta$, the
$k$th digit after the point equals the number of occurrences of the
$k$th biggest number $\mu_k$ in $e$. The base $\beta$ is chosen in such a way
that even if $e$ is the union of some of the given multisets, the
number $\nu(x,e)$ of occurrences of $x$ in $e$ never exceeds
$\beta-1$.
Therefore, the following sequence
of implications holds:

\begin{align*}
  a^*>_{lex}b^*\text{ and }c^*>_{lex}d^*
  &\Rightarrow \#(a)>\#(b)\text{ and }\#(c)>\#(d)\\
  &\Rightarrow \#(a)+\#(c)>\#(b)+\#(d)\\
  &\Rightarrow \#(a\mcup c)>\#(b\mcup d)\\
  &\Rightarrow (a\mcup c)^*>_{lex}(b\mcup d)^*
\end{align*}
\end{proof}

Consequently, the corresponding three relations on collections induced by
(\ref{eq:prefto}) are all semi-linear comparison relations.

As a natural example of an irreflexive and transitive relation on multisets of reals that does not satisfy
the monotonicity condition (\ref{eq:mon1}) consider $\succ_{av}$ defined by
\[
\mbox{
$a \succ_{av} b$ iff $(\sum_{i=1}^{m} a_i)/m > (\sum_{j=1}^{n} b_j)/n$.
}
\]
Note that for
\[
a := \mso3\msc, b := \mso2,2,2,2\msc, c := \mso1,1,1,1\msc, d := \mso0\msc
\]
we have both $a \succ_{av} b$ and $c \succ_{av} d$ but not
$a \mcup c \succ_{av} b \mcup d$ since $\mso3,1,1,1,1\msc\succ_{av} \mso2,2,2,2,0\msc$
does not hold.

Further, the following natural irreflexive and transitive relations on multisets of reals do not satisfy
the monotonicity condition (\ref{eq:mon2}):

\begin{itemize}

\item the \emph{elitist} order:

$a \succ_{el} b$ iff $max(a) > max(b)$,

\item the \emph{egalitarian} order:

$a \succ_{eg} b$ iff $min(a) > min(b)$,

\end{itemize}

Indeed, we have both $\mso2\msc \succ_{el} \mso1\msc$ and $\mso2\msc \succ_{eg} \mso1\msc$, but neither
$\mso3,2\msc \succ_{el} \mso3,2\msc$ nor $\mso1,0\msc \succ_{eg} \mso1,0\msc$ holds.

\section{Individual values}
\label{sec:ind}

In the previous section we defined the comparison relations in the
context of TU-games by comparing the values (yielded by the $v$
function) of whole coalitions. Alternatively, we could compare payoffs
to individual players. The idea is that in the end, the value secured
by a coalition may have to be distributed to its members, and this
final payoff to a player may determine his preferences.

To formalize this approach we need the notion of an \emph{individual value function} $\phi$ that,
given the $v$ function of a TU-game and a coalition $A$, assigns to
each player $i \in A$ a real value $\phi^v_i(A)$. We assume
that $\phi$ is \emph{efficient}, i.e.~that it exactly distributes
the coalition's value to its members:
\begin{equation*}
  \sum_{i \in A} \phi^v_i(A) = v(A).
\end{equation*}

For a collection $C := \{C_1, \dots, C_k\}$, we put
\begin{equation*}
  \phi^v(C) := \mso \phi^v_i(A) \mid A\in C, i \in A\msc.
\end{equation*}

Given two collections $C=\{C_1,\dots,C_k\}$ and
$C'=\{C'_1,\dots,C'_l\}$ with $\bigcup C=\bigcup C'$, the comparison
relations now compare $\phi^v(C)$ and $\phi^v(C')$, which are
multisets of $|\bigcup C|$ real numbers, one
number for each player. In this way it is guaranteed that the
comparison relations are \emph{anonymous} in the sense that the names
of the players do not play a role.

In this section, to distinguish between comparison relations defined
only by means of $v$ and those defined using both $v$ and $\phi$, we
denote the former by $\prefto_v$ and the latter by $\prefto_\phi$.

We now examine how these two different approaches for
defining comparison relations relate. To this end,
we will clarify when they coincide, i.e.~when given a comparison
relation defined in one way, we can also obtain it using the other
way, and when they are unrelated.
We begin by formalizing the concept of anonymity.

\begin{definition}
Assume a coalitional TU-game $(N,v)$.

  \begin{itemize}
  \item An individual value function $\phi$ is \emph{anonymous} if for all $v$, permutations $\pi$ of $N$,
  $i\in N$, and $A\subseteq N$
  \begin{equation*}
    \phi^v_i(A)=\phi^{v\circ\pi^{-1}}_{\pi(i)}(\pi(A)).
  \end{equation*}

\item   $v$ is \emph{anonymous} if for all permutations $\pi$ of $N$ and
  $A\subseteq N$
  \begin{equation*}
    v(A) = v(\pi(A)).
  \end{equation*}
  \end{itemize}
\end{definition}

Note that for all $A$ we have $(v\circ\pi^{-1})(\pi(A)) = v(A)$.
Intuitively, $\phi$ is anonymous if it does not depend on the names of
the players and $v$ is anonymous if it is defined only in terms of the
cardinality of the argument coalition.

The following simple observation holds.

\begin{note}
  For any $v$ and $\phi$, if $\prefto_v$ and $\prefto_\phi$ are the
  utilitarian order (as defined in Section~\ref{sec:TU}), then
  for all collections $C$ and $C'$, we have
  $\phi^v(C)\prefto_\phi\phi^v(C')$ iff $v(C)\prefto_vv(C')$.
\end{note}
\begin{proof}
  Immediate since
  \begin{equation*}
    \sum_{A\in C}v(A)=\sum_{A\in C}\sum_{i\in
      A}\phi^v_i(A)=\sum_{A\in C,i\in A}\phi^v_i(A).
  \end{equation*}
\end{proof}

For other orders discussed in Section \ref{sec:TU} no relation between 
$\prefto_v$ and $\prefto_\phi$ holds. In fact, we have the following results.

\begin{theorem}
  Given $v$ and $\prefto_v$, it is in general not possible to define
  an anonymous individual value function $\phi$ along with $\prefto_\phi$
  such that for all collections $C$ and $C'$, we have
  $\phi^v(C)\prefto_\phi\phi^v(C')$ iff $v(C)\prefto_vv(C')$. This
  holds even if we restrict ourselves to anonymous $v$.
\end{theorem}

\begin{proof}
Consider the following game with $N=\{1,2\}$:
  \begin{align*}
    v(\{1\})&:=1 & v(\{2\})&:=1 & v(\{1,2\})&:=2,
  \end{align*}
  and take $\prefto_v$ to be the Nash order as defined in
  Section~\ref{sec:TU}. This yields both
  \begin{equation*}
    v(\{\{1,2\}\})=\mso2\msc\prefto_v\mso1,1\msc=v(\{\{1\},\{2\}\})
  \end{equation*}
  and
  \begin{equation*}
    v(\{\{1\},\{2\}\})\nprefto_vv(\{\{1,2\}\}).
  \end{equation*}
  However, the symmetry of the game and anonymity of $\phi$ forces
  \begin{equation*}
    \phi^v(\{\{1,2\}\})=\mso1,1\msc=\phi^v(\{\{1\},\{2\}\}),
  \end{equation*}
  so we have either
  \begin{equation*}
    \phi^v(\{\{1,2\}\})\prefto_\phi\phi^v(\{\{1\},\{2\}\})\text{ and }
    \phi^v(\{\{1\},\{2\}\})\prefto_\phi\phi^v(\{\{1,2\}\})
  \end{equation*}
  or
  \begin{equation*}
    \phi^v(\{\{1,2\}\})\nprefto_\phi\phi^v(\{\{1\},\{2\}\})\text{ and }
    \phi^v(\{\{1\},\{2\}\})\nprefto_\phi\phi^v(\{\{1,2\}\}).
  \end{equation*}
\end{proof}

\begin{theorem}
  Given $v, \: \phi$ and $\prefto_\phi$, it is in general not possible
  to define $\prefto_v$ such that for all collections $C$ and $C'$, we
  have $v(C)\prefto_vv(C')$ iff $\phi^v(C)\prefto_\phi\phi^v(C')$.
  This holds even if we restrict ourselves to anonymous $v$,
  anonymous $\phi$, and a Nash or leximin order
  (as defined in Section~\ref{sec:TU}) for $\prefto_\phi$.
\end{theorem}
\begin{proof}
  Consider  $N=\{1,\dots,4\}$ and
  \begin{align*}
    v(A)&:=6\text{ for all $A\subseteq N$}\\
    \phi^v_i(A)&:=\frac{v(A)}{|A|}.
  \end{align*}
  Then we have
  \begin{align*}
    \phi^v(\{\{1\},\{2,3,4\}\})&=\mso6,2,2,2\msc\\
    \phi^v(\{\{1,2\},\{3,4\}\})&=\mso3,3,3,3\msc,
  \end{align*}
  which are distinguished by each of the mentioned $\prefto_\phi$, while
  \begin{align*}
    v(\{\{1\},\{2,3,4\}\})&=v(\{\{1,2\},\{3,4\}\})=\mso6,6\msc.
  \end{align*}
\end{proof}

These results suggest that the two approaches for defining
preference relations are fundamentally different and coincide only for
the utilitarian order.

In the case of individual values we can introduce natural orders that
have no counterpart for the comparison relations defined
only by means of $v$.  The reason is that for each partition $P$,
$\phi^v(P)$ can be alternatively viewed as a sequence (of payoffs) of
(the same) length $n$. Such sequences can then be compared using

\begin{itemize}

\item the \emph{majority order}:

$(k_1, \LL, k_n)  \succ_{m} (l_1, \LL, l_n)$ iff $|\C{i \mid k_i > l_i}| > |\C{i \mid l_i > k_i}|$,

\item the \emph{Pareto order}:

$(k_1, \LL, k_n)  \succ_{p} (l_1, \LL, l_n)$ iff 

$\fa i \in \{1,\ldots ,n\} \: k_i \geq l_i$ 
and $\te i \in \{1,\ldots ,n\} \: k_i > l_i$.

\end{itemize}

The relation $\succ_{m}$ is clearly irreflexive and monotonic both in
sense~\eqref{eq:mon1} and~\eqref{eq:mon2}. Unfortunately, it is not
transitive. Indeed, we have both $(2, 3, 0) \succ_m (1,2,2)$ and
$(1,2,2) \succ_m (3,1,1)$, but $(2, 3, 0) \succ_m (3,1,1)$ does not
hold.  In contrast, the relation $\succ_{p}$ is transitive, irreflexive,
monotonic both in sense~\eqref{eq:mon1} and~\eqref{eq:mon2}.

\section{Stable partitions}
\label{sec:stable}

We now return to our analysis of partitions.  One way to identify
conditions guaranteeing the unique outcome of the iterations of the
merge and split rules is through focusing on the properties of such a
unique outcome.  This brings us to the concept of a stable partition.

We follow here the approach of \cite{AR06}, although now no notion of a
game is present.  The introduced notion is parametrized by means of a
\emph{defection function} \defect{} that assigns to each partition
some partitioned subsets of the grand coalition.  Intuitively, given a
partition $P$ the family $\mathbb{D}(P)$ consists of all the
collections $C := \{C_1,\ldots, C_l\}$ whose players can leave the
partition $P$ by forming a new, separate, group of players
$\cup_{j=1}^l C_j$ divided according to the collection $C$.  Two most
natural defection functions are $\mathbb{D}_p$, which allows formation
of all partitions of the grand coalition, and $\mathbb{D}_c$, which
allows formation of all collections in the grand coalition.

Next, given a collection $C$ and a partition $P := \{P_1,\ldots, P_k\}$ we define
\[
C[P] := \{P_1 \cap \bigcup C, \ldots,P_k \cap \bigcup C\} \setminus \{\ES\}
\]
and call $C[P]$
the \textit{collection} $C$ \textit{in the frame of} $P$.
(By removing the empty set we ensure that $C[P]$ is a collection.)
To clarify this concept consider Figure \ref{fig:coa}. We
depict in it a collection $C$, a partition $P$ and $C$ in the frame of
$P$ (together with $P$). Here $C$ consists of four coalitions, while
$C$ in the frame of $P$ consists of three coalitions.

\begin{figure}[htbp]
\centering%
\beginpgfgraphicnamed{coll-frame}%
\begin{tikzpicture}[thick,scale=0.75,transform shape]
  \begin{scope}[shape=ellipse,minimum width=6cm,minimum height=2cm]
    \draw (0,0) node[draw] (P) {} ;
    \draw (0,-3) node[draw] (P2) {} ;
  \end{scope}

  \begin{scope}[shape=ellipse,minimum width=3cm,minimum
    height=0.9cm,fill opacity=0.3,fill=lightgray]
    \draw (-0.1,2.5) node[draw,fill] (C) {} ;
    \draw (-0.1,-3) node[draw,fill] (C2) {} ;
  \end{scope}

  \draw (C.15) -- (C.165) ;
  \draw (C.5) -- (C.175) ;
  \draw (C.-10) -- (C.190) ;

  \draw (P.north west) -- (P.south west) ;
  \draw (P.140) -- (P.-140) ;
  \draw (P.95) -- (P.-95) ;
  \draw (P.30) -- (P.-30) ;

  \draw (P2.north west) -- (P2.south west) ;
  \draw (P2.140) -- (P2.-140) ;
  \draw (P2.95) -- (P2.-95) ;
  \draw (P2.30) -- (P2.-30) ;

  \node[below] at (C.south) {Collection $C$} ;
  \node[below] at (P.south) {Partition $P$} ;
  \node[below] at (P2.south) {$C[P]$} ;
\end{tikzpicture}%
\endpgfgraphicnamed%
\caption{A collection $C$ in the frame of a partition $P$}
\label{fig:coa}
\end{figure}
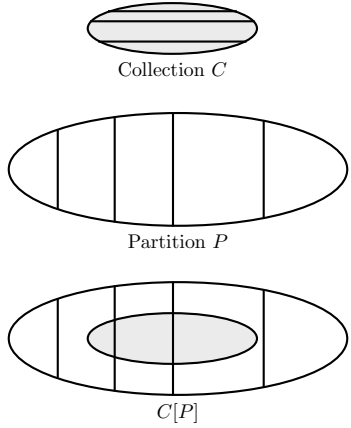

Intuitively, given a subset $S$ of $N$ and a partition $C :=
\{C_1,\ldots, C_l\}$ of $S$, the collection $C$ offers the players
from $S$ the `benefits' resulting from the partition of $S$ by $C$.
However, if a partition $P$ of $N$ is `in force', then the players
from $S$ enjoy instead the benefits resulting from the partition of
$S$ by $C[P]$, i.e.~$C$ in the frame of $P$.

To get familiar with the $C[P]$ notation note that

\begin{itemize}
\item if $C$
is a singleton, say $C = \{T\}$, then
$\{T\}[P] = \{P_1 \cap  T, \LL, P_k \cap T\} \setminus \{\ES\}$,
where $P = \{P_1, \LL, P_k\}$,

\item if $C$ is a partition of $N$, then $C[P] = P$,

\item if $C \sse P$,
that is $C$ consists of some coalitions of $P$, then $C[P] = C$.
\end{itemize}

In general the following simple observation holds.

\begin{note} \label{not:cp}
For a collection $C$ and a partition $P$, $C[P] = C$ iff
each element of $C$ is a subset of a different element of $P$.
  \HB
\end{note}



This brings us to the following notion.

\begin{definition}
Assume a defection function $\defect$ and a comparison relation
$\prefto$.
We call a partition $P$ \defect-stable if $C[P]\prefto C$
for all $C\in\defect(P)$ such that  $C[P] \neq C$.
\end{definition}

The last qualification, that is  $C[P] \neq C$, requires some
explanation. Intuitively, this condition indicates that the players
only care about the way they are partitioned.  Indeed, if $C[P] = C$,
then the partitions of $\bigcup C$ by means of $P$ and by means of $C$
coincide and are viewed as equally satisfactory for the players in
$\bigcup C$.  By disregarding the situations in which $C[P] = C$ we
therefore adopt a limited viewpoint of cooperation according to which
the players in $C$ do not care about the presence of the players from
outside of $\bigcup C$ in their coalitions.

The following observation holds, where we call a partition $P$ of $N$
\emph{$\prefto$-maximal} if for all partitions $P'$ of $N$ different from $P$, $P \prefto P'$ holds.

\begin{theorem} \label{thm:Dp}
A partition of $N$ is $\mathbb{D}_p$-stable iff it is
$\prefto$-maximal.

In particular, a $\mathbb{D}_p$-stable partition of $N$ exists if
$\prefto$ is a semi-linear comparison relation.
\HB
\end{theorem}
\begin{proof}
Note that if $C$ is a partition of $N$,
then $C[P] \neq C$ is equivalent to the statement $P \neq C$, since
then $C[P] = P$. 
So a partition $P$ of $N$ is $\mathbb{D}_p$-stable iff
for all partitions $P'\neq P$ of $N$, $P \prefto P'$ holds.
\end{proof}
\VV

In contrast, $\mathbb{D}_c$-stable partitions do not need to exist even
if the comparison relation $\prefto$ is semi-linear. 

\begin{example}
Consider $N=\{1,2,3\}$ and any semi-linear
comparison relation $\prefto$ such that
$\{\{1,2,3\}\} \prefto\{\{1\},\{2\},\{3\}\}$ and
$\{\{a\},\{b\}\}\prefto\{\{a,b\}\}$
for all $a, b \in \{1,2,3\}$, $a \neq b$.


Then no partition of $N$ is $\mathbb{D}_c$-stable.
Indeed, $P := \{\{1\},\{2\},\{3\}\}$ is not $\mathbb{D}_c$-stable since
for $C := \{\{1,2,3\}\}$ we have
$C[P] = \{\{1\},\{2\},\{3\}\}\nprefto\{\{1,2,3\}\} = C$.
Further, any other
partition $P$ contains some coalition $\{a,b\}$ and is thus not $\mathbb{D}_c$-stable 
either since then for $C := \{\{a\},\{b\}\}$ we have
\[
C[P] = \{\{a,b\}\}\nprefto\{\{a\},\{b\}\} = C.
\]
\HB
\end{example}

In \cite{AR06} another example is given for the case of TU-games and utilitarian order.
More precisely, a TU-game is defined in which no
$\mathbb{D}_c$-stable partition exists, where $\prefto$ is defined through
(\ref{eq:prefto}) using the utilitarian order $\succ_{ut}$.

\section{Stable partitions and merge/split rules}
\label{sec:unique}

We now resume our investigation of the conditions under which every
iteration of the merge and split rules yields the same outcome.  
To establish the main theorem of the paper and provide an answer in terms of the
$\mathbb{D}_c$-stable partitions, we first present the following three lemmata about
$\defect_c$-stable partitions.

\begin{lemma}
  \label{not:stable}
Every $\defect_c$-stable partition is closed under the applications of
the merge and split rules.
\end{lemma}

\begin{proof}
To prove the closure of a $\defect_c$-stable partition $P$
under the merge rule assume that for some $\{T_1,\dots,T_k\}\subseteq P$ we have
  $\{\bigcup_{j=1}^k  T_j\} \prefto\{T_1,\dots,T_k\}$.
$\defect_c$-stability of $P$ with
$C := \{\bigcup_{j=1}^k  T_j\}$
yields
  \begin{equation*}
    \{T_1,\dots,T_k\}= \{\bigcup_{j=1}^k
      T_j\}[P]\prefto \{\bigcup_{j=1}^k T_j\},
  \end{equation*}
  which is a contradiction by virtue of the transitivity and
  irreflexivity of $\prefto$.

The closure under the split rule is shown analogously.
\end{proof}

Next, we provide a characterization of 
$\defect_c$-stable partitions. Given a
partition $P := \C{P_1, \ldots, P_k}$ we call here a coalition $T$
$P$-\textit{compatible} if for some $i \in \{1, \ldots, k\}$ we have
$T \sse P_i$ and $P$-\textit{incompatible} otherwise.

\begin{lemma}
  \label{theorem:3.1}
A partition $P=\{P_1,\dots,P_k\}$ of $N$ is $\defect_c$-stable
  iff the following two conditions are satisfied (see
  Figure~\ref{fig:comp-incomp} for an illustration of the following coalitions):
  \begin{enumerate}[(i)]
  \item for each $i\in\{1,\dots,k\}$ and each pair of disjoint
    coalitions $A$ and $B$ such that $A\cup B\subseteq P_i$
    \begin{equation}
      \label{eq:3.1:3}
      \{A\cup B\}\prefto\{A,B\},
    \end{equation}
  \item for each $P$-incompatible coalition $T\subseteq N$
    \begin{equation}
      \label{eq:3.1:4}
      \{T\}[P]\prefto\{T\}.
    \end{equation}
  \end{enumerate}
\end{lemma}

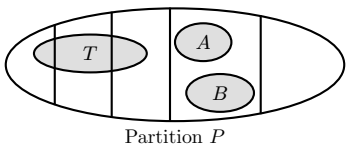
\begin{figure}[htbp]
  \centering%
  \beginpgfgraphicnamed{comp-incomp}%
  \begin{tikzpicture}[thick,scale=0.75,transform shape]
    \begin{scope}[shape=ellipse,minimum width=6cm,minimum height=2cm]
      \draw (0,0) node[draw] (P) {} ;
    \end{scope}
    
    \begin{scope}[shape=ellipse,fill opacity=0.25,fill=gray]
      \draw (.5,.4) node[minimum width=1cm,minimum height=0.6cm,draw,fill] (A) {$A$} ;
      \draw (.8,-.5) node[minimum width=1.2cm,minimum height=0.4cm,draw,fill] (B) {$B$} ;
      \draw (-1.5,.2) node[minimum width=2cm,minimum height=0.4cm,draw,fill] (T) {$T$} ;
    \end{scope}
    
    \draw (P.north west) -- (P.south west) ;
    \draw (P.140) -- (P.-140) ;
    \draw (P.95) -- (P.-95) ;
    \draw (P.30) -- (P.-30) ;
    
    \node[below] at (P.south) {Partition $P$} ;
    \node at (A.center) {$A$} ;
    \node at (B.center) {$B$} ;
    \node at (T.center) {$T$} ;
  \end{tikzpicture}%
  \endpgfgraphicnamed%
  \caption{$P$-compatible coalitions $A$ and $B$ and
    a $P$-incompatible coalition $T$ as in Lemma~\ref{theorem:3.1}}
  \label{fig:comp-incomp}
\end{figure}

\begin{proof}
  ($\Rightarrow$) It suffices to note that for $C = \{A, B\}$ we have
$C[P] = \{A \cup B\}$ and for $C = \{T\}$ we have  $\{T\}[P] \neq \{T\}$
by the $P$-incompatibility of  $T$. Then $(i)$ and $(ii)$ follow directly by the definition
of $\defect_c$-stability.
\\[2mm]
  \NI
($\Leftarrow$)
Transitivity, monotonicity \eqref{eq:mon2} and
  \eqref{eq:3.1:3} imply by induction that for each
  $i\in\{1,\dots,k\}$ and each collection $C=\{C_1,\dots,C_l\}$ with
  $l>1$ and $\bigcup C\subseteq P_i$,
  \begin{equation}
    \label{eq:3.1:5}
    \left\{\bigcup C\right\}\prefto C.
  \end{equation}

  Let now $C$ be an arbitrary collection in $N$ such that $C[P] \neq
  C$. We prove that $C[P]\prefto C$. Define
\[
\mbox{$D^i:=\{T\in C \mid T\subseteq P_i\}$,}
\]
\[
\mbox{$E:=C\setminus\bigcup_{i=1}^k D^i$,}
\]
\[
\mbox{$E^i:=\{P_i\cap T \mid T\in E\}\setminus\{\emptyset\}$. }
\]

Note that $D^i$ is the set of $P$-compatible elements of $C$
contained in $P_i$, $E$ is
the set of $P$-incompatible elements of $C$ and $E^i$ consists of the
non-empty intersections of $P$-incompatible elements of $C$ with $P_i$.

Suppose now that 
$\bigcup_{i=1}^k E^i \neq \ES$. Then $E \neq \ES$ and consequently
  \begin{equation}
    \bigcup_{i=1}^kE^i
    =\bigcup_{i=1}^k\left(\{P_i\cap T \mid T\in E\}\setminus\{\emptyset\}\right)
    =\bigcup_{T\in E}\left(\{T\}[P]\right)
    \overset{\eqref{eq:mon1},\eqref{eq:3.1:4}}{\prefto}E.
\label{eq:E}
  \end{equation}

Consider now the following property:

\begin{equation}
  \label{eq:prop}
|D^i\cup E^i|>1.
\end{equation}
Fix $i\in\{1,\dots,k\}$.
If (\ref{eq:prop}) holds, then 
  \begin{equation*}
\left\{P_i\cap\bigcup C\right\}=\left\{\bigcup (D^i \cup E^i)\right\}\overset{\eqref{eq:3.1:5}}{\prefto} D^i\cup E^i
  \end{equation*}
and otherwise
  \begin{equation*}
\left\{P_i\cap\bigcup C\right\}= \left\{D^i\cup E^i\right\}.
  \end{equation*}

Recall now that 
\begin{equation*}
    C[P]=\bigcup_{i=1}^k\left\{P_i\cap\bigcup C\right\}\setminus\{\emptyset\}.
\end{equation*}

We distinguish two cases.
\II

\NI
\emph{Case 1}. (\ref{eq:prop}) holds for some $i\in\{1,\dots,k\}$.

Then by \eqref{eq:mon1} and \eqref{eq:mon2}
  \begin{equation*}
    C[P] \prefto \bigcup_{i=1}^k(D^i\cup E^i)=(C\setminus E)\cup\bigcup_{i=1}^kE^i.
  \end{equation*}

If $\bigcup_{i=1}^kE^i=\emptyset$, then also $E=\emptyset$ and we 
get $C[P] \prefto C$. Otherwise by (\ref{eq:E}), transitivity and
\eqref{eq:mon2}
  \begin{equation*}
    C[P]\prefto(C\setminus E)\cup E=C.
  \end{equation*}

\NI
\emph{Case 2}. (\ref{eq:prop}) does not hold 
for any $i\in\{1,\dots,k\}$.

Then 
  \begin{equation*}
    C[P] = \bigcup_{i=1}^k(D^i\cup E^i)=(C\setminus E)\cup\bigcup_{i=1}^kE^i.
  \end{equation*}
  Moreover, because $C[P] \neq C$, by Note~\ref{not:cp} a
  $P$-incompatible element in $C$ exists. So $\bigcup_{i=1}^k E^i \neq
  \ES$ and by (\ref{eq:E}) and \eqref{eq:mon2} we get as before
  \begin{equation*}
    C[P]\prefto(C\setminus E)\cup E=C.
  \end{equation*}
\end{proof}

In \cite{AR06} the above characterization was proved for the coalitional
TU-games and the utilitarian order.
We shall now use it in the proof of the following lemma.

\begin{lemma}
  \label{lemma:unique}
  Assume that $P$   is $\defect_c$-stable.
  Let $P'$ be closed under applications of merge and split rules. Then
  $P'=P$.
\end{lemma}
\begin{proof}
  Suppose $P=\{P_1,\dots,P_k\}$, $P'=\{T_1,\dots,T_m\}$. Assume $P\neq
  P'$. Then there is $i_0\in\{1,\dots,k\}$ such that for all
  $j\in\{1,\dots,m\}$ we have $P_{i_0}\neq T_j$. Let
  $T_{j_1},\dots,T_{j_l}$ be the minimum cover of
  $P_{i_0}$. 
In the following case distinction we use Lemma~\ref{theorem:3.1}.
\II

\NI
\emph{Case 1}. $P_{i_0}=\bigcup_{h=1}^lT_{j_h}$.

    Then $\{T_{j_1},\dots,T_{j_l}\}$ is a proper partition of
    $P_{i_0}$. But \eqref{eq:3.1:3} (through its generalization to
(\ref{eq:3.1:5})) yields
    $P_{i_0}\prefto \{T_{j_1},\dots,T_{j_l}\}$, thus the merge rule is
    applicable to $P'$.
\II

\NI
\emph{Case 2}. $P_{i_0}\subsetneq\bigcup_{h=1}^lT_{j_h}$.

    Then for some $j_h$ we have $\emptyset\neq P_{i_0}\cap
    T_{j_h}\subsetneq T_{j_h}$, so $T_{j_h}$ is $P$-incompatible. By
    \eqref{eq:3.1:4} we have $\{T_{j_h}\}[P] \prefto \{T_{j_h}\}$, thus
    the split rule is applicable to $P'$.
\end{proof}
\VV

We can now present the desired result.

\begin{theorem} \label{thm:outcome}
Suppose that $\prefto$ is a comparison relation and
$P$ is a $\mathbb{D}_c$-stable partition.
Then
\begin{enumerate}[(i)]

\item $P$ is the outcome of every iteration of the merge and split rules.

\item $P$ is a unique $\mathbb{D}_p$-stable partition.

\item $P$ is a unique $\mathbb{D}_c$-stable partition.

\end{enumerate}

\end{theorem}

\begin{proof}
\NI
$(i)$ By Note \ref{not:1} every iteration of the merge and split
  rules terminates,
  so the claim  follows by Lemma \ref{lemma:unique}. \\[2mm]
  $(ii)$ Since $P$ is $\mathbb{D}_c$-stable, it is in particular
  $\mathbb{D}_p$-stable. By
Theorem \ref{thm:Dp} for all partitions $P'\neq P$, $P \prefto P'$ holds.
So uniqueness follows from the transitivity and
  irreflexivity of $\prefto$. \\[2mm]
  $(iii)$ Suppose that $P'$ is a $\mathbb{D}_c$-stable partition.  By
  Lemma \ref{not:stable} $P'$ is closed under the applications of the
  merge and split rules, so by Lemma \ref{lemma:unique} $P' = P$.
\end{proof}

This theorem generalizes~\cite{AR06}, where this result
was established for the coalitional TU-games and the 
utilitarian order.  It was also shown there that there exist
coalitional TU-games in which all iterations of the merge and split
rules have a unique outcome which is not a $\mathbb{D}_c$-stable
partition.

\section{Applications}
\label{sec:conc}

The obtained results do not involve
any notion of a game. 
In this section we show applications to three classes of coalitional games.
In each case we define a class of games and a natural comparison relation
for which all iterations of the merge and split
rules have a unique outcome.

\subsection{Coalitional TU-games}

To show that the obtained results naturally apply to coalitional
TU-games consider first the special case of the utilitarian order,
according to which given a coalitional TU-game $(N, v)$,
for two collections $P := \{P_1, \LL, P_k\}$ and $Q
= \{Q_1, \LL, Q_l\}$ such that $\bigcup P = \bigcup Q$, we have
\[
\mbox{$P \prefto Q$ iff $\sum_{i=1}^k v(P_i) > \sum_{i=1}^l v(Q_i)$.}
\]

Recall that $(N, v)$ is called
\emph{strictly superadditive} if for each pair of disjoint coalitions $A$ and $B$
\[
\mbox{$v(A) + v(B) < v(A\cup B)$.}
\]

Further, recall from \cite[ page 241]{Owe01} that given a partition $P := \{P_1, \LL, P_k\}$ of $N$ and
coalitional TU-games $(P_1, v_1), \LL, (P_k, v_k)$, 
their \emph{composition} $(N, \oplus_{i=1}^k v_i)$
is defined by
\[
(\oplus_{i=1}^k v_i)(A) = \sum_{i=1}^k v_i(P_i \cap A).
\]

We now modify this definition and introduce the concept of a \emph{semi-union} of  $(P_1, v_1), \LL, (P_k, v_k)$, 
written as $(N, \overline{\oplus}_{i=1}^k v_i)$, and defined by
\[
(\overline{\oplus}_{i=1}^k v_i)(A) := \left\{ 
\begin{tabular}{ll}
$(\oplus_{i=1}^k v_i)(A)$ &  \mbox{if $A \sse P_i$ for some $i$} \\
$(\oplus_{i=1}^k v_i)(A) - \epsilon$ &  \mbox{otherwise}
\end{tabular}
\right . 
\]
where $\epsilon > 0$.

So for $P$-incompatible coalitions the payoff is
strictly smaller for the semi-union of TU-games than for their union,
while for other coalitions the payoffs are the same.
It is then easy to prove using Lemma~\ref{theorem:3.1} that in the
semi-union $(N, \overline{\oplus}_{i=1}^k v_i)$ of strictly superadditive
TU-games the partition $P$ is $\defect_c$-stable. Consequently, by
Theorem~\ref{thm:outcome}, in this game $P$ is the outcome of every
iteration of the merge and split rules.

The following more general example deals with arbitrary
monotonic comparison relations as introduced in Sections~\ref{sec:TU} and \ref{sec:ind}.

\begin{example}
  \label{ex:genericstable}
  Given a partition $P := \{P_1, \ldots, P_k\}$ of $N$, with $\prefto$
  being one of the orders defined in Section~\ref{sec:TU}, we 
  define a TU-game for which $P$ is the outcome of every iteration of the merge and split rules.

Let
  \begin{align*}
    f(x,y)&:=
    \begin{cases}
      x+y&\text{if $\prefto$ is the utilitarian order}\\
      x\cdot y&\text{if $\prefto$ is the Nash order}\\
      \max\{x,y\}&\text{if $\prefto$ is the leximin order}\\
    \end{cases}\\
    \intertext{and define}
    v(A)&:= 
    \begin{cases}
      1&\text{if $|A|=1$}\\
      \max\limits_{\substack{B\cup C=A\\\text{$B$,$C$
            disjoint}\\\text{coalitions}}}\{f(v(B),v(C))\} + 1&\text{if $|A|>1$ and $A\subseteq P_i$
        for some $i$}\\
      0&\text{otherwise.}
  \end{cases}
  \end{align*}
  Then
  \begin{enumerate}[(i)]
  \item for any two disjoint coalitions $A,B$ with $A\cup B\subseteq P_i$
    for some $i$, we have \[ v(A\cup B)>f(v(A),v(B)) \] by construction of
    $v$, and thus
    \begin{itemize}
    \item $v(A\cup B)>v(A)+v(B)$ for utilitarian $\prefto$;
    \item $v(A\cup B)>v(A)\cdot v(B)$ for Nash $\prefto$;
    \item $v(A\cup B)>\max\{v(A),v(B)\}$ for leximin $\prefto$.
    \end{itemize}
    Hence in all cases $\{A\cup B\}\prefto\{A,B\}$.
  \item for any $P$-incompatible coalition $T\subseteq N$, we have
    \[ v(A)>0 \text{ for all $A\in\{T\}[P]$,\quad and } v(T)=0.
    \]
    Hence $\{T\}[P]\prefto\{T\}$.
  \end{enumerate}
  Lemma~\ref{theorem:3.1} now implies that $P$ is indeed $\defect_c$-stable,
  so Theorem~\ref{thm:outcome} applies.
\HB
\end{example}

\begin{example}
  Given a partition $P := \{P_1, \ldots, P_k\}$ of $N$, with $\prefto$
  being one of the orders defined in Section~\ref{sec:TU} or Pareto order of
  Section~\ref{sec:ind}, we define a TU-game and an individual
  value function for which $P$ is the outcome of every iteration of the merge and split rules.

Let

  \begin{align*}
    f(x,y)&:=
    \begin{cases}
      \vert N\vert\cdot\max\{x,y\}+1 &\text{if $\prefto$ is
        leximin or Pareto}\\
      x+y &\text{otherwise,}
    \end{cases}\\
    \intertext{define $v$ as in Example~\ref{ex:genericstable},
      and define}
    \phi_i^v(A)&:=\frac{v(A)}{\vert A\vert}.
  \end{align*}
  Then
  \begin{enumerate}[(i)]
  \item for any two disjoint coalitions $A,B$ with $A\cup B\subseteq P_i$
    for some $i$, we have \[ v(A\cup B)>f(v(A),v(B)) \] again by construction of
    $v$, and thus
    \begin{itemize}
    \item for utilitarian or Nash $\prefto$:\\
      $v(A\cup B)>v(A)+v(B)$, and since $\phi_i^v$ distributes the
      value evenly, in all cases $\{A\cup B\}\prefto\{A,B\}$,
    \item for leximin or Pareto $\prefto$:\\
      $v(A\cup B)>\vert A\cup B\vert\cdot\max\{v(A),v(B)\}$,\\
      thus $\phi_i^v(A\cup B)>\max\{v(A),v(B)\}$ for all $i$,\\
      thus $\{A\cup B\}\prefto\{A,B\}$ in all cases,
    \end{itemize}
  \item for any $P$-incompatible coalition $T\subseteq N$,
    $\{T\}[P]\prefto\{T\}$ as before.
  \end{enumerate}
  Again, Lemma~\ref{theorem:3.1} implies that $P$ is $\defect_c$-stable,
and Theorem~\ref{thm:outcome} applies.
\HB
\end{example}

\subsection{Hedonic games}

Recall that a \textit{hedonic game} $(N, \succeq_1, \LL, \succeq_n)$
consists of a set of players $N = \C{1, \LL, n}$ and a sequence of
linear preorders $\succeq_1, \LL, \succeq_n$, where each $\succeq_i$
is the preference of player $i$ over the subsets of $N$ containing
$i$.  In what follows we shall not need the assumption that the
$\succeq_i$ relations are linear.  Denote by $\succ_i$ the associated
irreflexive relation.

Given a partition $P$ of $N$ and player $i$ we denote by $P(i)$ the
element of $P$ to which $i$ belongs and call it the set of
\emph{friends of $i$ in $P$}.

We now provide an example of a hedonic game in which a
$\mathbb{D}_c$-stable partition w.r.t.~to a natural comparison
relation $\succ$ exists.

To this end we assume that, given a partition $P := \{P_1, \LL, P_k\}$
of $N$, each player

\begin{itemize}
\item prefers a larger set of his friends in $P$ over a smaller one,

\item `dislikes' coalitions that include a player who is not his friend in $P$.
\end{itemize}

We formalize this by putting for all sets of players that
include $i$
\[
\mbox{$S \succeq_i T$ iff $T \sse S \sse P(i)$,}
\]
and by extending this order to the coalitions that include player
$i$ and possibly players from outside of $P(i)$ by assuming that such coalitions are
the minimal elements in $\succeq_i$.
So 
\[
\mbox{$S \succ_i T$ iff either $T \subsetneq S \sse P(i)$ or $S \sse
  P(i)$ and not $T \sse P(i)$.}
\]

We then define for two partitions $Q$ and $Q'$ of the same set of players
\[
\mbox{$Q \prefto Q'$ iff for $i \in \C{1,\LL,n}$ $Q(i) \succeq_i Q'(i)$ with at least one $\succeq_i$ being strict.}
\]

(Note the similarity between this relation and the $\succ_{p}$ relation introduced in Section \ref{sec:ind}.)
It is straightforward to check that $\prefto$ is indeed a
comparison relation
and that the partition $P$ satisfies then conditions
(\ref{eq:3.1:3}) and (\ref{eq:3.1:4}) of Lemma \ref{theorem:3.1}.
So by virtue of this lemma $P$ is $\defect_c$-stable.
Consequently, on the account of Theorem \ref{thm:outcome},
the partition $P$ is the outcome of every
iteration of the merge and split rules.

\subsection{Exchange economy games}

Recall that an \emph{exchange economy} consists of

\begin{itemize}
\item a market with $k$ goods,

\item for each player $i$ an initial endowment of these goods represented
by a vector $\vec \omega_i \in {\cal R}_+^k$,

\item for each player $i$ a transitive and linear relation $\succeq_i$ using which he can compare the bundles of the goods,
represented as vectors from ${\cal R}_+^k$.
\end{itemize}

An \emph{exchange economy game} is then defined by first taking as the set of \emph{outcomes}
the set of all sequences of bundles,
\[
X := \C{(\vec {x}_1, \LL, \vec {x}_n)  \mid \vec {x}_i \in {\cal R}_+^k \mbox{ for } i \in N},
\]
i.e. $X = ({\cal R}_+^k)^n$, 
and extending each preference relation $\succeq_i$ from the set ${\cal R}_+^k$ of all bundles
to the set $X$ by putting for $\vec {x}, \vec {y} \in X$
\begin{equation}
  \label{eq:global}
\mbox{$\vec {x} \succeq_i \vec {y}$ iff $\vec {x}_i \succeq_i \vec {y}_i$.}  
\end{equation}
This simply means that each player is only interested in his own bundle.

Then we assign to each coalition $S$ the following set of outcomes:
\[
\mbox{$V(S) := \{\vec {x} \in X \mid  \sum_{i \in S} \vec {x}_i = \sum_{i \in S} \vec \omega_i$  and $\vec {x}_j = \vec \omega_j$ for all $j \in N \setminus S$\}.}
\]
So $V(S)$ consists of the set of outcomes that can be achieved by trading between the members of $S$.

Given a partition $P=\{P_1,\dots,P_k\}$ of $N=\{1,\dots,n\}$
we now define a specific exchange economy game with $n$ goods (one type of good for each player) as follows, where
$i\in N$:
\begin{align*}
  \vec \omega_i&:=\text{characteristic vector of $P(i)$},\\
  \vec x_i&\succeq_i \vec y_i\text{ iff }x_{i,i}\geq y_{i,i}\quad\text{and}\quad\vec x_i\succ_i \vec y_i\text{ iff }x_{i,i}>y_{i,i},
\end{align*}
that is, each player's initial endowment consists of exactly one good of the type
of each of his friends in $P$, and he prefers a bundle if he gets more goods
of his own type.


Now let $A\prefto B$ iff
\begin{gather*}
  \forall A_l\in A \setminus B \: \exists\vec x\in
  V(A_l)\: \forall j\in A_l\\
  \left[\big(\forall\vec y\in V(B(j))\vec
    x\succ_j\vec y\big)\vee\big(\forall\vec y\in V(B(j))\vec x\succeq_j\vec
    y\wedge |A_l|<|B(j)|\big)\right].
\end{gather*}

So a partition $A$ is preferred to a partition $B$ if each coalition $A_l$
of $A$ not present in the partition $B$
can achieve an outcome which each player of $A_l$
strictly prefers to any outcome of his respective coalition in $B$,
or which he likes at least as much as any outcome of his respective coalition
in $B$ when that coalition is strictly larger than $A_l$.
The intuition is that the players' preferences over outcomes weigh
most, but in case of ties the players prefer smaller coalitions.

It is easy to check that $\prefto$ is a comparison relation.
We now prove that the partition $P$ is $\mathbb{D}_c$-stable w.r.t.~$\prefto$.
First, note that by the definition of the initial endowments
for all $l\in\{1,\dots,k\}$ and coalitions $A\subseteq P_l$ there is an outcome
$\vec z_A\in V(A)$ which gives exactly $|A|$ units of good $j$ to each player
$j\in A$. We have $\vec z_A\succeq_i\vec x$ for all $i\in A$ and $\vec
x\in V(A)$. This implies that $P$ is $\mathbb{D}_c$-stable by Lemma
\ref{theorem:3.1} since
\begin{enumerate}[(i)]
\item for each pair of disjoint
    coalitions $A$ and $B$ such that $A\cup B\subseteq P_l$
we have $\vec z_{A\cup B}\succ_i\vec z_A$ for each $i\in A$ and $\vec z_{A\cup
    B}\succ_i\vec z_B$ for each $i\in B$ since $|A\cup B|>|A|$ and
  $|A\cup B|>|B|$, thus $\{A\cup B\}\prefto\{A,B\}$,

\item for any $P$-incompatible $T\subseteq N$, $A\in\{T\}[P]$,
  $i\in A$, and $\vec x\in V(T)$, we have $\vec z_A\succeq_i\vec x$
  (since player $i$ can get in $T$ at most all goods of his
  type from his friends in $P$, which are exactly the same as in $A$), and
  $|A|<|T|$, thus $\{T\}[P]\prefto\{T\}$.
\end{enumerate}

Consequently, in the above game, by Theorem \ref{thm:outcome}
the partition $\{P_1, \LL, P_k\}$ is the outcome of every
iteration of the merge and split rules.

\section{Conclusions}
\label{sec:conclusions}

We have presented a generic approach to coalition formation,
in which the only possible operations on coalitions are merges and splits.
These operations can take place when they result in an improvement
with respect to some given comparison relation on partitions of the involved subset of players.
Such a comparison relation needs to satisfy only a few natural properties,
namely irreflexivity, transitivity and monotonicity,
and we have given examples induced by several well-known orders in the context of TU-games.

We have identified natural conditions under which every iteration of
merges and splits yields a unique outcome, which led to a natural
notion of a stable partition.  We have shown that besides TU-games our
approach and results also naturally apply to hedonic games and
exchange economy games.

It would be interesting to extend this approach and allow other transformations,
such as transfers (moving a subset of one coalition to another)
or, more generally, swaps (exchanging subsets of two coalitions),
as considered in~\cite{AR06} in the setting of TU-games and utilitarian order.

\section*{Acknowledgements}
We thank Tadeusz Radzik for helpful comments.

\bibliographystyle{plain}

\bibliography{e}

\begin{thebibliography}{10}

\bibitem{Apt04}
K.~R. Apt.
\newblock Uniform proofs of order independence for various strategy elimination
  procedures.
\newblock {\em The B.E. Journal of Theoretical Economics, 4(1)}, 2004.
\newblock (Contributions), Article 5, 48 pages. Available from
  \url{http://xxx.lanl.gov/abs/cs.GT/0403024}.

\bibitem{AR06}
K.~R. Apt and T.~Radzik.
\newblock Stable partitions in coalitional games, 2006.
\newblock Available from \verb+http://arxiv.org/abs/cs.GT/0605132+.

\bibitem{AD74}
R.J. Aumann and J.H. Dr\`{e}ze.
\newblock Cooperative games with coalition structures.
\newblock {\em International Journal of Game Theory}, 3:217--237, 1974.

\bibitem{BJ06}
F.~Bloch and M.~Jackson.
\newblock Definitions of equilibrium in network formation games.
\newblock {\em International Journal of Game Theory}, (34):305--318, 2006.

\bibitem{BJ02}
A.~Bogomolnaia and M.~Jackson.
\newblock The stability of hedonic coalition structures.
\newblock {\em Games and Economic Behavior}, 38(2):201--230, 2002.

\bibitem{BZ03}
N.~Burani and W.S. Zwicker.
\newblock Coalition formation games with separable preferences.
\newblock {\em Mathematical Social Sciences}, 45(1):27--52, 2003.

\bibitem{Dem04}
G.~Demange.
\newblock On group stability in hierarchies and networks.
\newblock {\em Journal of Political Economy}, 112(4):754--778, 2004.

\bibitem{DW05}
G.~Demange and M.~Wooders, editors.
\newblock {\em Group Formation in Economics}.
\newblock Cambridge University Press, 2006.

\bibitem{Gre94}
J.~Greenberg.
\newblock Coalition structures.
\newblock In R.J. Aumann and S.~Hart, editors, {\em Handbook of Game Theory
  with Economic Applications}, volume~2 of {\em Handbook of Game Theory with
  Economic Applications}, chapter~37, pages 1305--1337. Elsevier, 1994.

\bibitem{MSPCP06}
I.~Macho-Stadler, D.~P\'{e}rez-Castrillo, and N.~Porteiro.
\newblock Sequential formation of coalitions through bilateral agreements in a
  cournot setting.
\newblock {\em International Journal of Game Theory}, (34):207--228, 2006.

\bibitem{RePEc:urb:wpaper:07_12}
M.~Marini.
\newblock An overview of coalition \& network formation models for economic
  applications.
\newblock Working Papers 0712, University of Urbino Carlo Bo, Department of
  Economics, April 2007.
\newblock available at http://ideas.repec.org/p/urb/wpaper/07\_12.html.

\bibitem{Mou88}
H.~Moulin.
\newblock {\em Axioms of Cooperative Decision Making}.
\newblock Cambridge University Press, 1998.

\bibitem{Owe01}
G.~Owen.
\newblock {\em Game Theory}.
\newblock Academic Press, New York, third edition, 2001.

\bibitem{RV97}
D.~Ray and R.~Vohra.
\newblock Equilibrium binding agreements.
\newblock {\em Journal of Economic Theory}, (73):30--78, 1997.

\bibitem{SBK01}
T.~S\"{o}nmez, S.~Banerjee, and H.~Konishi.
\newblock Core in a simple coalition formation game.
\newblock {\em Social Choice and Welfare}, 18(1):135--153, 2001.

\bibitem{Ter03}
Terese.
\newblock {\em Term Rewriting Systems}.
\newblock Cambridge Tracts in Theoretical Computer Science 55. Cambridge
  University Press, 2003.

\bibitem{Yi97}
S.S. Yi.
\newblock Stable coalition structures with externalities.
\newblock {\em Games and Economic Behavior}, 20:201--237, 1997.

\end{thebibliography}

\end{document}